\documentclass[11pt]{article}
\usepackage[a4paper]{geometry}
\usepackage{amsfonts, amsmath, amssymb, amsthm, graphicx, caption, authblk, multirow, makecell, framed, float, xcolor, enumitem, tikz, hyperref}
\usepackage{rotating}
\definecolor{blk}{RGB}{63,63,63}
\theoremstyle{definition}
\newtheorem{theorem}{Theorem}

\theoremstyle{remark}

\newcommand*{\mybox}[1]{%
  \framebox{\raisebox{0cm}[0.5\baselineskip][0.05\baselineskip]{%
    \hbox to 0.10cm {\hss#1\hss}}}\hspace{0.05cm}}

\begin{document}
\title{Wataridori is NP-Complete}
\author[1]{Suthee Ruangwises\thanks{\texttt{suthee@cp.eng.chula.ac.th}}}
\affil[1]{Department of Computer Engineering, Faculty of Engineering, Chulalongkorn University, Bangkok, Thailand}
\date{}
\maketitle

\begin{abstract}
Wataridori is a pencil puzzle that involves drawing paths in a rectangular grid to connect circles into pairs while satisfying several constraints. In this paper, we prove that deciding whether a given Wataridori puzzle has a solution is NP-complete via a reduction from Numberlink, another pencil puzzle that has previously been proved NP-complete.

\textbf{Keywords:} NP-hardness, computational complexity, Wataridori, Numberlink, puzzle
\end{abstract}

\section{Introduction}
\textit{Wataridori} is a pencil puzzle introduced by Nikoli \cite{janko}, a Japanese publisher that developed many famous pencil puzzles including Sudoku, Kakuro, and Slitherlink. The puzzle consists of a rectangular grid partitioned into polyominoes called \textit{regions}. (A polyomino is a shape formed by one or more unit squares connected orthogonally.) Some cells contain a circle, and some of these circles contain a number.

The objective of Wataridori is to connect all circles into pairs using paths that go between horizontally or vertically adjacent cells. Paths cannot cross or share a cell with one another. The two circles connected by a path must contain the same number, which must be equal to the total number of regions the path passes through. A path may enter and exit a region at most once. Circles without numbers are treated as wildcards and may represent any number. See Figure~\ref{fig1}. Note that it is not required for all cells in the grid to be covered by paths.

\begin{figure}
\centering
\begin{tikzpicture}
\draw[step=1cm,color={rgb:black,1;white,4}] (0,0) grid (6,6);

\draw[line width=0.6mm] (0,0) -- (0,6);
\draw[line width=0.6mm] (1,1) -- (1,2);
\draw[line width=0.6mm] (1,4) -- (1,5);
\draw[line width=0.6mm] (2,0) -- (2,1);
\draw[line width=0.6mm] (2,2) -- (2,6);
\draw[line width=0.6mm] (3,0) -- (3,5);
\draw[line width=0.6mm] (4,0) -- (4,3);
\draw[line width=0.6mm] (4,4) -- (4,5);
\draw[line width=0.6mm] (5,1) -- (5,6);
\draw[line width=0.6mm] (6,0) -- (6,6);

\draw[line width=0.6mm] (0,0) -- (6,0);
\draw[line width=0.6mm] (0,1) -- (3,1);
\draw[line width=0.6mm] (4,1) -- (6,1);
\draw[line width=0.6mm] (1,2) -- (6,2);
\draw[line width=0.6mm] (2,3) -- (4,3);
\draw[line width=0.6mm] (0,4) -- (1,4);
\draw[line width=0.6mm] (3,4) -- (6,4);
\draw[line width=0.6mm] (1,5) -- (5,5);
\draw[line width=0.6mm] (0,6) -- (6,6);

\node[draw,circle,minimum size=0.6cm] at (0.5,0.5) {2};
\node[draw,circle,minimum size=0.6cm] at (0.5,4.5) {};
\node[draw,circle,minimum size=0.6cm] at (0.5,5.5) {};
\node[draw,circle,minimum size=0.6cm] at (1.5,1.5) {};
\node[draw,circle,minimum size=0.6cm] at (1.5,2.5) {5};
\node[draw,circle,minimum size=0.6cm] at (1.5,3.5) {2};
\node[draw,circle,minimum size=0.6cm] at (3.5,2.5) {2};
\node[draw,circle,minimum size=0.6cm] at (3.5,3.5) {};
\node[draw,circle,minimum size=0.6cm] at (3.5,4.5) {3};
\node[draw,circle,minimum size=0.6cm] at (3.5,5.5) {5};
\node[draw,circle,minimum size=0.6cm] at (4.5,1.5) {5};
\node[draw,circle,minimum size=0.6cm] at (4.5,3.5) {2};
\node[draw,circle,minimum size=0.6cm] at (4.5,5.5) {3};
\node[draw,circle,minimum size=0.6cm] at (5.5,4.5) {8};
\end{tikzpicture}
\hspace{1cm}
\begin{tikzpicture}
\draw[step=1cm,color={rgb:black,1;white,4}] (0,0) grid (6,6);

\draw[line width=0.6mm] (0,0) -- (0,6);
\draw[line width=0.6mm] (1,1) -- (1,2);
\draw[line width=0.6mm] (1,4) -- (1,5);
\draw[line width=0.6mm] (2,0) -- (2,1);
\draw[line width=0.6mm] (2,2) -- (2,6);
\draw[line width=0.6mm] (3,0) -- (3,5);
\draw[line width=0.6mm] (4,0) -- (4,3);
\draw[line width=0.6mm] (4,4) -- (4,5);
\draw[line width=0.6mm] (5,1) -- (5,6);
\draw[line width=0.6mm] (6,0) -- (6,6);

\draw[line width=0.6mm] (0,0) -- (6,0);
\draw[line width=0.6mm] (0,1) -- (3,1);
\draw[line width=0.6mm] (4,1) -- (6,1);
\draw[line width=0.6mm] (1,2) -- (6,2);
\draw[line width=0.6mm] (2,3) -- (4,3);
\draw[line width=0.6mm] (0,4) -- (1,4);
\draw[line width=0.6mm] (3,4) -- (6,4);
\draw[line width=0.6mm] (1,5) -- (5,5);
\draw[line width=0.6mm] (0,6) -- (6,6);

\draw[line width=0.6mm,color=red] (0.5,0.8) -- (0.5,3.5) -- (1.2,3.5);
\draw[line width=0.6mm,color=red] (0.5,4.8) -- (0.5,5.2);
\draw[line width=0.6mm,color=red] (1.5,1.2) -- (1.5,0.5) -- (5.5,0.5) -- (5.5,4.2);
\draw[line width=0.6mm,color=red] (1.8,2.5) -- (2.5,2.5) -- (2.5,1.5) -- (4.2,1.5);
\draw[line width=0.6mm,color=red] (3.8,2.5) -- (4.5,2.5) -- (4.5,3.2);
\draw[line width=0.6mm,color=red] (3.8,4.5) -- (4.5,4.5) -- (4.5,5.2);
\draw[line width=0.6mm,color=red] (3.2,3.5) -- (2.5,3.5) -- (2.5,4.5) -- (1.5,4.5) -- (1.5,5.5) -- (3.2,5.5);

\node[draw,circle,minimum size=0.6cm] at (0.5,0.5) {2};
\node[draw,circle,minimum size=0.6cm] at (0.5,4.5) {};
\node[draw,circle,minimum size=0.6cm] at (0.5,5.5) {};
\node[draw,circle,minimum size=0.6cm] at (1.5,1.5) {};
\node[draw,circle,minimum size=0.6cm] at (1.5,2.5) {5};
\node[draw,circle,minimum size=0.6cm] at (1.5,3.5) {2};
\node[draw,circle,minimum size=0.6cm] at (3.5,2.5) {2};
\node[draw,circle,minimum size=0.6cm] at (3.5,3.5) {};
\node[draw,circle,minimum size=0.6cm] at (3.5,4.5) {3};
\node[draw,circle,minimum size=0.6cm] at (3.5,5.5) {5};
\node[draw,circle,minimum size=0.6cm] at (4.5,1.5) {5};
\node[draw,circle,minimum size=0.6cm] at (4.5,3.5) {2};
\node[draw,circle,minimum size=0.6cm] at (4.5,5.5) {3};
\node[draw,circle,minimum size=0.6cm] at (5.5,4.5) {8};
\end{tikzpicture}
\caption{An example of a $6 \times 6$ Wataridori puzzle (left) and its solution (right)}
\label{fig1}
\end{figure}

\subsection{Related Work}
Many pencil puzzles have been proved NP-complete, where the decision problem asks whether a given instance has a valid solution. Examples include Chained Block \cite{chained}, Choco Banana \cite{choco}, Dosun-Fuwari \cite{dosun}, Fillmat \cite{fillmat}, Five Cells \cite{fivecells}, Goishi Hiroi \cite{goishi}, Hashiwokakero \cite{bridges}, Herugolf \cite{makaro}, Heyawake \cite{heyawake}, Juosan \cite{juosan}, Kakuro \cite{sudoku}, Kurodoko \cite{kurodoko}, Kurotto \cite{juosan}, LITS \cite{lits}, Makaro \cite{makaro}, Moon-or-Sun \cite{moon}, Nagareru \cite{moon}, Nondango \cite{nondango}, Nonogram \cite{nonogram}, Norinori \cite{lits}, Nurikabe \cite{nurikabe}, Nurimeizu \cite{moon}, Nurimisaki \cite{nurimisaki}, Nuritwin \cite{nuritwin}, Pencils \cite{pencils}, Ripple Effect \cite{ripple}, Roma \cite{roma}, Sashigane \cite{nurimisaki}, Shakashaka \cite{shakashaka}, Shikaku \cite{ripple}, Slitherlink \cite{sudoku}, Sto-Stone \cite{stone}, Sudoku \cite{sudoku}, Suguru \cite{suguru}, Sumplete \cite{sumplete}, Tatamibari \cite{tatamibari}, Tilepaint \cite{sudoku}, Toichika \cite{toichika}, Usowan \cite{usowan}, Yin-Yang \cite{yinyang}, Yosenabe \cite{yosenabe}, and Zeiger \cite{zeiger}.

The vast majority of NP-completeness proofs for pencil puzzles are obtained via reductions from SAT-type problems, such as 3SAT, 1-in-3 SAT, Planar SAT, and Circuit SAT. Some proofs instead use reductions from Hamiltonian cycle problems. In contrast, direct reductions from one pencil puzzle to another are rare, as the rules of different puzzles typically differ significantly, making such reductions difficult to construct.

\subsection{Our Contribution}
In this paper, we prove that deciding solvability of a given Wataridori puzzle is NP-complete via direct reductions from Numberlink, another pencil puzzle that has already been proved to be NP-complete.

\section{Numberlink}
\textit{Numberlink} (also known as \textit{Arukone}) is a well-known pencil puzzle. While it was popularized by Nikoli, its origins date back to at least 1932 \cite{janko2}. Despite its appearance being similar to Wataridori, its rules are significantly different.

The puzzle consists of a rectangular grid in which some cells contain numbers, with each number appearing exactly twice. The objective is to connect every pair of cells with the same number by a path that goes between horizontally or vertically adjacent cells. Paths cannot cross or share cells with one another. See Figure~\ref{fig3}. Note that, unlike Wataridori, the numbers in Numberlink have no quantitative meaning and merely serve as labels to distinguish pairs of cells (in some Numberlink mobile apps, letters or colors are used instead of numbers).

\begin{figure}
\centering
\begin{tikzpicture}
\draw[step=1cm] (0,0) grid (6,6);

\node at (0.5,2.5) {4};
\node at (0.5,3.5) {3};
\node at (0.5,4.5) {2};
\node at (0.5,5.5) {1};
\node at (1.5,1.5) {5};
\node at (3.5,4.5) {1};
\node at (4.5,1.5) {5};
\node at (4.5,2.5) {3};
\node at (4.5,5.5) {2};
\node at (5.5,5.5) {4};
\end{tikzpicture}
\hspace{1cm}
\begin{tikzpicture}
\draw[step=1cm] (0,0) grid (6,6);

\draw[line width=0.6mm,color=red] (0.8,5.5) -- (3.5,5.5) -- (3.5,4.8);
\draw[line width=0.6mm,color=red] (0.8,4.5) -- (2.5,4.5) -- (2.5,3.5) -- (4.5,3.5) -- (4.5,5.2);
\draw[line width=0.6mm,color=red] (0.8,3.5) -- (1.5,3.5) -- (1.5,2.5) -- (4.2,2.5);
\draw[line width=0.6mm,color=red] (0.5,2.2) -- (0.5,0.5) -- (5.5,0.5) -- (5.5,5.2);
\draw[line width=0.6mm,color=red] (1.8,1.5) -- (4.2,1.5);

\node at (0.5,2.5) {4};
\node at (0.5,3.5) {3};
\node at (0.5,4.5) {2};
\node at (0.5,5.5) {1};
\node at (1.5,1.5) {5};
\node at (3.5,4.5) {1};
\node at (4.5,1.5) {5};
\node at (4.5,2.5) {3};
\node at (4.5,5.5) {2};
\node at (5.5,5.5) {4};
\end{tikzpicture}
\caption{An example of a $6 \times 6$ Numberlink puzzle (left) and its solution (right)}
\label{fig3}
\end{figure}

There are two main variants of Numberlink: one that requires all cells in the grid to be covered by paths and one that does not. Both variants have been proved to be NP-complete \cite{numberlink,numberlink0}.

\section{Reduction from Numberlink to Wataridori}
As Wataridori clearly belongs to NP, the nontrivial part is to prove its NP-hardness. We do so by constructing a reduction from Numberlink to Wataridori. We consider the variant of Numberlink in which it is not required for all cells in the grid to be covered by paths.

\subsection{Constructing a Block Representing Each Cell}
Given an $m \times n$ Numberlink puzzle $G$, we transform it into a $(4k+5)m \times (4k+5)n$ Wataridori puzzle $H$, where $k$ will be defined later. For each cell of $G$, we construct a $(4k+5) \times (4k+5)$ block in $H$. A cell containing a number is represented by a \textit{number block}, while an empty cell is represented by an \textit{empty block}. Examples of a number block and an empty block for the case $k=2$ are shown in Figures~\ref{figA} and \ref{figB}.

The blocks are placed in $H$ in the same arrangement as the corresponding cells in $G$, resulting in a $(4k+5)m \times (4k+5)n$ Wataridori puzzle.

\subsection{Idea of Reduction}
Suppose $G$ contains $p$ pairs of numbers. We construct $H$ so that the $p$ paths in a solution of $G$ correspond to $p$ paths in $H$ that pass through different numbers of regions. This property allows us to label the $p$ pairs of endpoints in $H$ with distinct numbers, thereby forcing the pairings in $H$ to match exactly those in $G$.

Consider the main path in Figure~\ref{figA} (the one connecting to the center circle labeled $x$). There are at least three ways for the path to perform a ``zig-zag'' before exiting the block on the right side, as shown in Figure~\ref{figC}. In these three cases, the path passes through six, eight, and ten distinct regions, respectively. The rightmost region is the region in which the path will remain while traversing empty blocks, until it enters the number block at the other endpoint. This region is also the same as the outermost region of the number block at that endpoint.

For general $k$, there are at least $k+1$ different ways for the path to perform a ``zig-zag'' before exiting the block on each side. In these cases, the path passes through $2k+2, 2k+4, 2k+6, \ldots, 4k+2$ different regions, respectively. 

For each solution path in $G$, its corresponding path in $H$ can perform a ``zig-zag'' in the number blocks at both endpoints. Hence, the total number of regions it passes through can be $4k+3, 4k+5, 4k+7, \ldots, 8k+3$, giving $2k+1$ possible values.

Finally, we choose $k = \lceil (p-1)/2 \rceil$, so that $2k+1 \ge p$. Therefore, we can assign $p$ distinct numbers from the set $\{4k+3, 4k+5, 4k+7, \ldots, 8k+3\}$ to the center circles of the $p$ pairs of number blocks in $H$ corresponding to the $p$ pairs of numbered cells in $G$. Specifically, if the numbers in $G$ are $1,2,\ldots,p$, then we assign the number $4k+2i+1$ to the center circle of a number block in $H$ corresponding to a cell labeled $i$ in $G$ ($i=1,2,\ldots,p$). This forces the pairings in $H$ to be exactly the same as in $G$.

Since $2p \le mn$ (the total number of cells in $G$), we have $k \le p/2 \le mn$, so the size of $H$ is polynomial in $m$ and $n$. Therefore, the transformation can be performed in polynomial time.

\begin{figure}
\centering
\begin{tikzpicture}
\draw[step=1cm] (0,0) grid (1,1);
\node at (0.5,0.5) {$x$};
\draw[line width=0.6mm,color=red] (0.65,0.5) -- (1,0.5);
\end{tikzpicture}
\vspace{3cm}

\begin{tikzpicture}
\draw[step=1cm,color={rgb:black,1;white,4}] (0,0) grid (13,13);
\draw[step=1cm,line width=0.6mm] (1,6) grid (5,8);
\draw[step=1cm,line width=0.6mm] (8,5) grid (12,7);
\draw[step=1cm,line width=0.6mm] (5,1) grid (7,5);
\draw[step=1cm,line width=0.6mm] (6,8) grid (8,12);

\draw[line width=0.6mm] (0,0) -- (0,6);
\draw[line width=0.6mm] (0,7) -- (0,13);
\draw[line width=0.6mm] (6,0) -- (6,6);
\draw[line width=0.6mm] (6,7) -- (6,13);
\draw[line width=0.6mm] (7,0) -- (7,6);
\draw[line width=0.6mm] (7,7) -- (7,13);
\draw[line width=0.6mm] (13,0) -- (13,6);
\draw[line width=0.6mm] (13,7) -- (13,13);

\draw[line width=0.6mm] (0,0) -- (6,0);
\draw[line width=0.6mm] (7,0) -- (13,0);
\draw[line width=0.6mm] (0,6) -- (6,6);
\draw[line width=0.6mm] (7,6) -- (13,6);
\draw[line width=0.6mm] (0,7) -- (6,7);
\draw[line width=0.6mm] (7,7) -- (13,7);
\draw[line width=0.6mm] (0,13) -- (6,13);
\draw[line width=0.6mm] (7,13) -- (13,13);

\draw[line width=0.6mm,color=red] (0.5,1.8) -- (0.5,2.2);
\draw[line width=0.6mm,color=red] (0.5,3.8) -- (0.5,4.2);
\draw[line width=0.6mm,color=red] (0.5,7.8) -- (0.5,8.2);
\draw[line width=0.6mm,color=red] (0.5,9.8) -- (0.5,10.2);
\draw[line width=0.6mm,color=red] (0.5,11.8) -- (0.5,12.2);
\draw[line width=0.6mm,color=red] (4.5,1.8) -- (4.5,2.2);
\draw[line width=0.6mm,color=red] (4.5,3.8) -- (4.5,4.2);
\draw[line width=0.6mm,color=red] (5.5,7.8) -- (5.5,8.2);
\draw[line width=0.6mm,color=red] (5.5,9.8) -- (5.5,10.2);
\draw[line width=0.6mm,color=red] (5.5,11.8) -- (5.5,12.2);
\draw[line width=0.6mm,color=red] (7.5,0.8) -- (7.5,1.2);
\draw[line width=0.6mm,color=red] (7.5,2.8) -- (7.5,3.2);
\draw[line width=0.6mm,color=red] (7.5,4.8) -- (7.5,5.2);
\draw[line width=0.6mm,color=red] (8.5,8.8) -- (8.5,9.2);
\draw[line width=0.6mm,color=red] (8.5,10.8) -- (8.5,11.2);
\draw[line width=0.6mm,color=red] (12.5,0.8) -- (12.5,1.2);
\draw[line width=0.6mm,color=red] (12.5,2.8) -- (12.5,3.2);
\draw[line width=0.6mm,color=red] (12.5,4.8) -- (12.5,5.2);
\draw[line width=0.6mm,color=red] (12.5,8.8) -- (12.5,9.2);
\draw[line width=0.6mm,color=red] (12.5,10.8) -- (12.5,11.2);

\draw[line width=0.6mm,color=red] (0.8,0.5) -- (1.2,0.5);
\draw[line width=0.6mm,color=red] (2.8,0.5) -- (3.2,0.5);
\draw[line width=0.6mm,color=red] (4.8,0.5) -- (5.2,0.5);
\draw[line width=0.6mm,color=red] (8.8,0.5) -- (9.2,0.5);
\draw[line width=0.6mm,color=red] (10.8,0.5) -- (11.2,0.5);
\draw[line width=0.6mm,color=red] (0.8,5.5) -- (1.2,5.5);
\draw[line width=0.6mm,color=red] (2.8,5.5) -- (3.2,5.5);
\draw[line width=0.6mm,color=red] (4.8,5.5) -- (5.2,5.5);
\draw[line width=0.6mm,color=red] (8.8,4.5) -- (9.2,4.5);
\draw[line width=0.6mm,color=red] (10.8,4.5) -- (11.2,4.5);
\draw[line width=0.6mm,color=red] (1.8,8.5) -- (2.2,8.5);
\draw[line width=0.6mm,color=red] (3.8,8.5) -- (4.2,8.5);
\draw[line width=0.6mm,color=red] (7.8,7.5) -- (8.2,7.5);
\draw[line width=0.6mm,color=red] (9.8,7.5) -- (10.2,7.5);
\draw[line width=0.6mm,color=red] (11.8,7.5) -- (12.2,7.5);
\draw[line width=0.6mm,color=red] (1.8,12.5) -- (2.2,12.5);
\draw[line width=0.6mm,color=red] (3.8,12.5) -- (4.2,12.5);
\draw[line width=0.6mm,color=red] (7.8,12.5) -- (8.2,12.5);
\draw[line width=0.6mm,color=red] (9.8,12.5) -- (10.2,12.5);
\draw[line width=0.6mm,color=red] (11.8,12.5) -- (12.2,12.5);

\draw[line width=0.6mm,color=red] (6.8,6.5) -- (13,6.5);

\node[draw,circle,minimum size=0.6cm] at (0.5,0.5) {1};
\node[draw,circle,minimum size=0.6cm] at (0.5,1.5) {1};
\node[draw,circle,minimum size=0.6cm] at (0.5,2.5) {1};
\node[draw,circle,minimum size=0.6cm] at (0.5,3.5) {1};
\node[draw,circle,minimum size=0.6cm] at (0.5,4.5) {1};
\node[draw,circle,minimum size=0.6cm] at (0.5,5.5) {1};
\node[draw,circle,minimum size=0.6cm] at (0.5,7.5) {1};
\node[draw,circle,minimum size=0.6cm] at (0.5,8.5) {1};
\node[draw,circle,minimum size=0.6cm] at (0.5,9.5) {1};
\node[draw,circle,minimum size=0.6cm] at (0.5,10.5) {1};
\node[draw,circle,minimum size=0.6cm] at (0.5,11.5) {1};
\node[draw,circle,minimum size=0.6cm] at (0.5,12.5) {1};
\node[draw,circle,minimum size=0.6cm] at (5.5,0.5) {1};
\node[draw,circle,minimum size=0.6cm] at (4.5,1.5) {1};
\node[draw,circle,minimum size=0.6cm] at (4.5,2.5) {1};
\node[draw,circle,minimum size=0.6cm] at (4.5,3.5) {1};
\node[draw,circle,minimum size=0.6cm] at (4.5,4.5) {1};
\node[draw,circle,minimum size=0.6cm] at (5.5,5.5) {1};
\node[draw,circle,minimum size=0.6cm] at (5.5,7.5) {1};
\node[draw,circle,minimum size=0.6cm] at (5.5,8.5) {1};
\node[draw,circle,minimum size=0.6cm] at (5.5,9.5) {1};
\node[draw,circle,minimum size=0.6cm] at (5.5,10.5) {1};
\node[draw,circle,minimum size=0.6cm] at (5.5,11.5) {1};
\node[draw,circle,minimum size=0.6cm] at (5.5,12.5) {1};
\node[draw,circle,minimum size=0.6cm] at (7.5,0.5) {1};
\node[draw,circle,minimum size=0.6cm] at (7.5,1.5) {1};
\node[draw,circle,minimum size=0.6cm] at (7.5,2.5) {1};
\node[draw,circle,minimum size=0.6cm] at (7.5,3.5) {1};
\node[draw,circle,minimum size=0.6cm] at (7.5,4.5) {1};
\node[draw,circle,minimum size=0.6cm] at (7.5,5.5) {1};
\node[draw,circle,minimum size=0.6cm] at (7.5,7.5) {1};
\node[draw,circle,minimum size=0.6cm] at (8.5,8.5) {1};
\node[draw,circle,minimum size=0.6cm] at (8.5,9.5) {1};
\node[draw,circle,minimum size=0.6cm] at (8.5,10.5) {1};
\node[draw,circle,minimum size=0.6cm] at (8.5,11.5) {1};
\node[draw,circle,minimum size=0.6cm] at (7.5,12.5) {1};
\node[draw,circle,minimum size=0.6cm] at (12.5,0.5) {1};
\node[draw,circle,minimum size=0.6cm] at (12.5,1.5) {1};
\node[draw,circle,minimum size=0.6cm] at (12.5,2.5) {1};
\node[draw,circle,minimum size=0.6cm] at (12.5,3.5) {1};
\node[draw,circle,minimum size=0.6cm] at (12.5,4.5) {1};
\node[draw,circle,minimum size=0.6cm] at (12.5,5.5) {1};
\node[draw,circle,minimum size=0.6cm] at (12.5,7.5) {1};
\node[draw,circle,minimum size=0.6cm] at (12.5,8.5) {1};
\node[draw,circle,minimum size=0.6cm] at (12.5,9.5) {1};
\node[draw,circle,minimum size=0.6cm] at (12.5,10.5) {1};
\node[draw,circle,minimum size=0.6cm] at (12.5,11.5) {1};
\node[draw,circle,minimum size=0.6cm] at (12.5,12.5) {1};

\node[draw,circle,minimum size=0.6cm] at (1.5,0.5) {1};
\node[draw,circle,minimum size=0.6cm] at (2.5,0.5) {1};
\node[draw,circle,minimum size=0.6cm] at (3.5,0.5) {1};
\node[draw,circle,minimum size=0.6cm] at (4.5,0.5) {1};
\node[draw,circle,minimum size=0.6cm] at (8.5,0.5) {1};
\node[draw,circle,minimum size=0.6cm] at (9.5,0.5) {1};
\node[draw,circle,minimum size=0.6cm] at (10.5,0.5) {1};
\node[draw,circle,minimum size=0.6cm] at (11.5,0.5) {1};
\node[draw,circle,minimum size=0.6cm] at (1.5,5.5) {1};
\node[draw,circle,minimum size=0.6cm] at (2.5,5.5) {1};
\node[draw,circle,minimum size=0.6cm] at (3.5,5.5) {1};
\node[draw,circle,minimum size=0.6cm] at (4.5,5.5) {1};
\node[draw,circle,minimum size=0.6cm] at (8.5,4.5) {1};
\node[draw,circle,minimum size=0.6cm] at (9.5,4.5) {1};
\node[draw,circle,minimum size=0.6cm] at (10.5,4.5) {1};
\node[draw,circle,minimum size=0.6cm] at (11.5,4.5) {1};
\node[draw,circle,minimum size=0.6cm] at (1.5,8.5) {1};
\node[draw,circle,minimum size=0.6cm] at (2.5,8.5) {1};
\node[draw,circle,minimum size=0.6cm] at (3.5,8.5) {1};
\node[draw,circle,minimum size=0.6cm] at (4.5,8.5) {1};
\node[draw,circle,minimum size=0.6cm] at (8.5,7.5) {1};
\node[draw,circle,minimum size=0.6cm] at (9.5,7.5) {1};
\node[draw,circle,minimum size=0.6cm] at (10.5,7.5) {1};
\node[draw,circle,minimum size=0.6cm] at (11.5,7.5) {1};
\node[draw,circle,minimum size=0.6cm] at (1.5,12.5) {1};
\node[draw,circle,minimum size=0.6cm] at (2.5,12.5) {1};
\node[draw,circle,minimum size=0.6cm] at (3.5,12.5) {1};
\node[draw,circle,minimum size=0.6cm] at (4.5,12.5) {1};
\node[draw,circle,minimum size=0.6cm] at (8.5,12.5) {1};
\node[draw,circle,minimum size=0.6cm] at (9.5,12.5) {1};
\node[draw,circle,minimum size=0.6cm] at (10.5,12.5) {1};
\node[draw,circle,minimum size=0.6cm] at (11.5,12.5) {1};

\node[draw,circle,minimum size=0.6cm] at (6.5,6.5) {$x$};
\end{tikzpicture}
\caption{A $13 \times 13$ number block in the Wataridori puzzle $H$ for $k=2$ (bottom), representing a cell labeled $x$ in the Numberlink puzzle $G$ (top). Red lines showing one possible solution corresponding to a path going to the right of the cell in $G$.}
\label{figA}
\end{figure}

\begin{figure}
\centering
\begin{tikzpicture}
\draw[step=1cm] (0,0) grid (1,1);
\draw[line width=0.6mm,color=red] (0,0.5) -- (0.5,0.5) -- (0.5,1);
\end{tikzpicture}
\vspace{3cm}

\begin{tikzpicture}
\draw[step=1cm,color={rgb:black,1;white,4}] (0,0) grid (13,13);

\draw[line width=0.6mm] (0,0) -- (0,6);
\draw[line width=0.6mm] (0,7) -- (0,13);
\draw[line width=0.6mm] (6,0) -- (6,6);
\draw[line width=0.6mm] (6,7) -- (6,13);
\draw[line width=0.6mm] (7,0) -- (7,6);
\draw[line width=0.6mm] (7,7) -- (7,13);
\draw[line width=0.6mm] (13,0) -- (13,6);
\draw[line width=0.6mm] (13,7) -- (13,13);

\draw[line width=0.6mm] (0,0) -- (6,0);
\draw[line width=0.6mm] (7,0) -- (13,0);
\draw[line width=0.6mm] (0,6) -- (6,6);
\draw[line width=0.6mm] (7,6) -- (13,6);
\draw[line width=0.6mm] (0,7) -- (6,7);
\draw[line width=0.6mm] (7,7) -- (13,7);
\draw[line width=0.6mm] (0,13) -- (6,13);
\draw[line width=0.6mm] (7,13) -- (13,13);

\draw[line width=0.6mm,color=red] (0.5,1.8) -- (0.5,2.2);
\draw[line width=0.6mm,color=red] (0.5,3.8) -- (0.5,4.2);
\draw[line width=0.6mm,color=red] (0.5,7.8) -- (0.5,8.2);
\draw[line width=0.6mm,color=red] (0.5,9.8) -- (0.5,10.2);
\draw[line width=0.6mm,color=red] (0.5,11.8) -- (0.5,12.2);
\draw[line width=0.6mm,color=red] (5.5,1.8) -- (5.5,2.2);
\draw[line width=0.6mm,color=red] (5.5,3.8) -- (5.5,4.2);
\draw[line width=0.6mm,color=red] (5.5,7.8) -- (5.5,8.2);
\draw[line width=0.6mm,color=red] (5.5,9.8) -- (5.5,10.2);
\draw[line width=0.6mm,color=red] (5.5,11.8) -- (5.5,12.2);
\draw[line width=0.6mm,color=red] (7.5,0.8) -- (7.5,1.2);
\draw[line width=0.6mm,color=red] (7.5,2.8) -- (7.5,3.2);
\draw[line width=0.6mm,color=red] (7.5,4.8) -- (7.5,5.2);
\draw[line width=0.6mm,color=red] (7.5,8.8) -- (7.5,9.2);
\draw[line width=0.6mm,color=red] (7.5,10.8) -- (7.5,11.2);
\draw[line width=0.6mm,color=red] (12.5,0.8) -- (12.5,1.2);
\draw[line width=0.6mm,color=red] (12.5,2.8) -- (12.5,3.2);
\draw[line width=0.6mm,color=red] (12.5,4.8) -- (12.5,5.2);
\draw[line width=0.6mm,color=red] (12.5,8.8) -- (12.5,9.2);
\draw[line width=0.6mm,color=red] (12.5,10.8) -- (12.5,11.2);

\draw[line width=0.6mm,color=red] (0.8,0.5) -- (1.2,0.5);
\draw[line width=0.6mm,color=red] (2.8,0.5) -- (3.2,0.5);
\draw[line width=0.6mm,color=red] (4.8,0.5) -- (5.2,0.5);
\draw[line width=0.6mm,color=red] (8.8,0.5) -- (9.2,0.5);
\draw[line width=0.6mm,color=red] (10.8,0.5) -- (11.2,0.5);
\draw[line width=0.6mm,color=red] (0.8,5.5) -- (1.2,5.5);
\draw[line width=0.6mm,color=red] (2.8,5.5) -- (3.2,5.5);
\draw[line width=0.6mm,color=red] (4.8,5.5) -- (5.2,5.5);
\draw[line width=0.6mm,color=red] (8.8,5.5) -- (9.2,5.5);
\draw[line width=0.6mm,color=red] (10.8,5.5) -- (11.2,5.5);
\draw[line width=0.6mm,color=red] (1.8,7.5) -- (2.2,7.5);
\draw[line width=0.6mm,color=red] (3.8,7.5) -- (4.2,7.5);
\draw[line width=0.6mm,color=red] (7.8,7.5) -- (8.2,7.5);
\draw[line width=0.6mm,color=red] (9.8,7.5) -- (10.2,7.5);
\draw[line width=0.6mm,color=red] (11.8,7.5) -- (12.2,7.5);
\draw[line width=0.6mm,color=red] (1.8,12.5) -- (2.2,12.5);
\draw[line width=0.6mm,color=red] (3.8,12.5) -- (4.2,12.5);
\draw[line width=0.6mm,color=red] (7.8,12.5) -- (8.2,12.5);
\draw[line width=0.6mm,color=red] (9.8,12.5) -- (10.2,12.5);
\draw[line width=0.6mm,color=red] (11.8,12.5) -- (12.2,12.5);

\draw[line width=0.6mm,color=red] (0,6.5) -- (6.5,6.5) -- (6.5,13);

\node[draw,circle,minimum size=0.6cm] at (0.5,0.5) {1};
\node[draw,circle,minimum size=0.6cm] at (0.5,1.5) {1};
\node[draw,circle,minimum size=0.6cm] at (0.5,2.5) {1};
\node[draw,circle,minimum size=0.6cm] at (0.5,3.5) {1};
\node[draw,circle,minimum size=0.6cm] at (0.5,4.5) {1};
\node[draw,circle,minimum size=0.6cm] at (0.5,5.5) {1};
\node[draw,circle,minimum size=0.6cm] at (0.5,7.5) {1};
\node[draw,circle,minimum size=0.6cm] at (0.5,8.5) {1};
\node[draw,circle,minimum size=0.6cm] at (0.5,9.5) {1};
\node[draw,circle,minimum size=0.6cm] at (0.5,10.5) {1};
\node[draw,circle,minimum size=0.6cm] at (0.5,11.5) {1};
\node[draw,circle,minimum size=0.6cm] at (0.5,12.5) {1};
\node[draw,circle,minimum size=0.6cm] at (5.5,0.5) {1};
\node[draw,circle,minimum size=0.6cm] at (5.5,1.5) {1};
\node[draw,circle,minimum size=0.6cm] at (5.5,2.5) {1};
\node[draw,circle,minimum size=0.6cm] at (5.5,3.5) {1};
\node[draw,circle,minimum size=0.6cm] at (5.5,4.5) {1};
\node[draw,circle,minimum size=0.6cm] at (5.5,5.5) {1};
\node[draw,circle,minimum size=0.6cm] at (5.5,7.5) {1};
\node[draw,circle,minimum size=0.6cm] at (5.5,8.5) {1};
\node[draw,circle,minimum size=0.6cm] at (5.5,9.5) {1};
\node[draw,circle,minimum size=0.6cm] at (5.5,10.5) {1};
\node[draw,circle,minimum size=0.6cm] at (5.5,11.5) {1};
\node[draw,circle,minimum size=0.6cm] at (5.5,12.5) {1};
\node[draw,circle,minimum size=0.6cm] at (7.5,0.5) {1};
\node[draw,circle,minimum size=0.6cm] at (7.5,1.5) {1};
\node[draw,circle,minimum size=0.6cm] at (7.5,2.5) {1};
\node[draw,circle,minimum size=0.6cm] at (7.5,3.5) {1};
\node[draw,circle,minimum size=0.6cm] at (7.5,4.5) {1};
\node[draw,circle,minimum size=0.6cm] at (7.5,5.5) {1};
\node[draw,circle,minimum size=0.6cm] at (7.5,7.5) {1};
\node[draw,circle,minimum size=0.6cm] at (7.5,8.5) {1};
\node[draw,circle,minimum size=0.6cm] at (7.5,9.5) {1};
\node[draw,circle,minimum size=0.6cm] at (7.5,10.5) {1};
\node[draw,circle,minimum size=0.6cm] at (7.5,11.5) {1};
\node[draw,circle,minimum size=0.6cm] at (7.5,12.5) {1};
\node[draw,circle,minimum size=0.6cm] at (12.5,0.5) {1};
\node[draw,circle,minimum size=0.6cm] at (12.5,1.5) {1};
\node[draw,circle,minimum size=0.6cm] at (12.5,2.5) {1};
\node[draw,circle,minimum size=0.6cm] at (12.5,3.5) {1};
\node[draw,circle,minimum size=0.6cm] at (12.5,4.5) {1};
\node[draw,circle,minimum size=0.6cm] at (12.5,5.5) {1};
\node[draw,circle,minimum size=0.6cm] at (12.5,7.5) {1};
\node[draw,circle,minimum size=0.6cm] at (12.5,8.5) {1};
\node[draw,circle,minimum size=0.6cm] at (12.5,9.5) {1};
\node[draw,circle,minimum size=0.6cm] at (12.5,10.5) {1};
\node[draw,circle,minimum size=0.6cm] at (12.5,11.5) {1};
\node[draw,circle,minimum size=0.6cm] at (12.5,12.5) {1};

\node[draw,circle,minimum size=0.6cm] at (1.5,0.5) {1};
\node[draw,circle,minimum size=0.6cm] at (2.5,0.5) {1};
\node[draw,circle,minimum size=0.6cm] at (3.5,0.5) {1};
\node[draw,circle,minimum size=0.6cm] at (4.5,0.5) {1};
\node[draw,circle,minimum size=0.6cm] at (8.5,0.5) {1};
\node[draw,circle,minimum size=0.6cm] at (9.5,0.5) {1};
\node[draw,circle,minimum size=0.6cm] at (10.5,0.5) {1};
\node[draw,circle,minimum size=0.6cm] at (11.5,0.5) {1};
\node[draw,circle,minimum size=0.6cm] at (1.5,5.5) {1};
\node[draw,circle,minimum size=0.6cm] at (2.5,5.5) {1};
\node[draw,circle,minimum size=0.6cm] at (3.5,5.5) {1};
\node[draw,circle,minimum size=0.6cm] at (4.5,5.5) {1};
\node[draw,circle,minimum size=0.6cm] at (8.5,5.5) {1};
\node[draw,circle,minimum size=0.6cm] at (9.5,5.5) {1};
\node[draw,circle,minimum size=0.6cm] at (10.5,5.5) {1};
\node[draw,circle,minimum size=0.6cm] at (11.5,5.5) {1};
\node[draw,circle,minimum size=0.6cm] at (1.5,7.5) {1};
\node[draw,circle,minimum size=0.6cm] at (2.5,7.5) {1};
\node[draw,circle,minimum size=0.6cm] at (3.5,7.5) {1};
\node[draw,circle,minimum size=0.6cm] at (4.5,7.5) {1};
\node[draw,circle,minimum size=0.6cm] at (8.5,7.5) {1};
\node[draw,circle,minimum size=0.6cm] at (9.5,7.5) {1};
\node[draw,circle,minimum size=0.6cm] at (10.5,7.5) {1};
\node[draw,circle,minimum size=0.6cm] at (11.5,7.5) {1};
\node[draw,circle,minimum size=0.6cm] at (1.5,12.5) {1};
\node[draw,circle,minimum size=0.6cm] at (2.5,12.5) {1};
\node[draw,circle,minimum size=0.6cm] at (3.5,12.5) {1};
\node[draw,circle,minimum size=0.6cm] at (4.5,12.5) {1};
\node[draw,circle,minimum size=0.6cm] at (8.5,12.5) {1};
\node[draw,circle,minimum size=0.6cm] at (9.5,12.5) {1};
\node[draw,circle,minimum size=0.6cm] at (10.5,12.5) {1};
\node[draw,circle,minimum size=0.6cm] at (11.5,12.5) {1};
\end{tikzpicture}
\caption{A $13 \times 13$ empty block in the Wataridori puzzle $H$ for $k=2$ (bottom), representing an empty cell in the Numberlink puzzle $G$ (top). Red lines showing one possible solution corresponding to a path entering from the left, turning left, and exiting at the top of the cell in $G$.}
\label{figB}
\end{figure}

\begin{figure}
\centering
\begin{tikzpicture}
\draw[step=1cm,color={rgb:black,1;white,4}] (0,0) grid (7,2);

\draw[line width=0.6mm] (0,0) -- (0,1);
\draw[line width=0.6mm] (1,0) -- (1,1);
\draw[line width=0.6mm] (2,0) -- (2,2);
\draw[line width=0.6mm] (3,0) -- (3,2);
\draw[line width=0.6mm] (4,0) -- (4,2);
\draw[line width=0.6mm] (5,0) -- (5,2);
\draw[line width=0.6mm] (6,0) -- (6,2);
\draw[line width=0.6mm] (7,0) -- (7,1);

\draw[line width=0.6mm] (0,0) -- (1,0);
\draw[line width=0.6mm] (2,0) -- (6,0);
\draw[line width=0.6mm] (1,1) -- (7,1);
\draw[line width=0.6mm] (1,2) -- (7,2);

\draw[line width=0.6mm,color=red] (0.8,1.5) -- (7,1.5);

\node[draw,circle,minimum size=0.6cm] at (0.5,1.5) {$x$};
\node[draw,circle,minimum size=0.6cm] at (1.5,0.5) {1};
\node[draw,circle,minimum size=0.6cm] at (6.5,0.5) {1};
\end{tikzpicture}

\vspace{1cm}

\begin{tikzpicture}
\draw[step=1cm,color={rgb:black,1;white,4}] (0,0) grid (7,2);

\draw[line width=0.6mm] (0,0) -- (0,1);
\draw[line width=0.6mm] (1,0) -- (1,1);
\draw[line width=0.6mm] (2,0) -- (2,2);
\draw[line width=0.6mm] (3,0) -- (3,2);
\draw[line width=0.6mm] (4,0) -- (4,2);
\draw[line width=0.6mm] (5,0) -- (5,2);
\draw[line width=0.6mm] (6,0) -- (6,2);
\draw[line width=0.6mm] (7,0) -- (7,1);

\draw[line width=0.6mm] (0,0) -- (1,0);
\draw[line width=0.6mm] (2,0) -- (6,0);
\draw[line width=0.6mm] (1,1) -- (7,1);
\draw[line width=0.6mm] (1,2) -- (7,2);

\draw[line width=0.6mm,color=red] (0.8,1.5) -- (2.5,1.5) -- (2.5,0.5) -- (3.5,0.5) -- (3.5,1.5) -- (7,1.5);

\node[draw,circle,minimum size=0.6cm] at (0.5,1.5) {$x$};
\node[draw,circle,minimum size=0.6cm] at (1.5,0.5) {1};
\node[draw,circle,minimum size=0.6cm] at (6.5,0.5) {1};
\end{tikzpicture}

\vspace{1cm}

\begin{tikzpicture}
\draw[step=1cm,color={rgb:black,1;white,4}] (0,0) grid (7,2);

\draw[line width=0.6mm] (0,0) -- (0,1);
\draw[line width=0.6mm] (1,0) -- (1,1);
\draw[line width=0.6mm] (2,0) -- (2,2);
\draw[line width=0.6mm] (3,0) -- (3,2);
\draw[line width=0.6mm] (4,0) -- (4,2);
\draw[line width=0.6mm] (5,0) -- (5,2);
\draw[line width=0.6mm] (6,0) -- (6,2);
\draw[line width=0.6mm] (7,0) -- (7,1);

\draw[line width=0.6mm] (0,0) -- (1,0);
\draw[line width=0.6mm] (2,0) -- (6,0);
\draw[line width=0.6mm] (1,1) -- (7,1);
\draw[line width=0.6mm] (1,2) -- (7,2);

\draw[line width=0.6mm,color=red] (0.8,1.5) -- (2.5,1.5) -- (2.5,0.5) -- (3.5,0.5) -- (3.5,1.5) -- (4.5,1.5) -- (4.5,0.5) -- (5.5,0.5) -- (5.5,1.5) -- (7,1.5);

\node[draw,circle,minimum size=0.6cm] at (0.5,1.5) {$x$};
\node[draw,circle,minimum size=0.6cm] at (1.5,0.5) {1};
\node[draw,circle,minimum size=0.6cm] at (6.5,0.5) {1};
\end{tikzpicture}
\caption{Three ways for the path to exit the block on the right side, passing through six, eight, and ten distinct regions, respectively}
\label{figC}
\end{figure}

\subsection{Proof of Correctness}
\begin{theorem}
$G$ has a solution if and only if $H$ has a solution.
\end{theorem}

\begin{proof}
Suppose $G$ has a solution $S$, with paths $s_1,s_2,\ldots,s_p$ connecting the pairs of numbers $1,2,\ldots,p$, respectively. For each $i=1,2,\ldots,p$, we draw a path in $H$ connecting the pair of circles labeled $4k+2i+1$, following the corresponding path in $S$ and performing exactly $i-1$ ``zig-zags'' in the number blocks at the two endpoints combined. 

The direct path with no ``zig-zag'' passes through $4k+3$ regions, and each additional ``zig-zag'' increases the number of regions by two. Hence, the resulting path passes through exactly $4k+2i+1$ regions, satisfying the puzzle constraint. Therefore, $H$ has a solution.

Conversely, suppose $H$ has a solution $T$. In each number block, there must be a path connecting the center circle that exits the block through one of its four sides. In each empty block, at most one path can enter and exit through two of the four sides (or no path passes through the block). Consequently, the $p$ paths in $T$ induce $p$ paths in $G$ connecting the pairs of numbers $1,2,\ldots,p$. Therefore, $G$ has a solution.
\end{proof}

Hence, Wataridori is NP-complete.

\section{Future Work}
We proved the NP-completeness of Wataridori via a reduction from another puzzle, Numberlink. Although the reduction is intuitive, it is nontrivial. Future work includes applying similar techniques to prove the NP-completeness of other puzzles via direct reductions from puzzles already known to be NP-complete. 

In particular, one possible candidate is \textit{Mintonette}, a puzzle that also involves drawing paths to connect circles in a rectangular grid into pairs under different constraints. It would be interesting to investigate whether there exists a direct reduction from Numberlink or Wataridori to Mintonette.

Another interesting direction is to study the solution structure of Wataridori. In particular, it remains open whether Wataridori is ASP-complete, that is, whether the problem of deciding if another solution exists given one solution is NP-complete.

\subsubsection*{Acknowledgement}
This work was supported by grants for development of new faculty staff, Ratchadaphiseksomphot Fund, Chulalongkorn University.


\end{document}